\documentclass[12pt,epsfig,amsmath,amssymb,pstricks]{article}
\usepackage{amsthm}
\usepackage{graphicx}
\usepackage{xcolor}
\newtheorem{Theorem}{Theorem}
\newtheorem{Lemma}{Lemma}

\begin{document}

\title{\bf Sharing nonfungible information requires \\
shared nonfungible information}

\author{{Chris Fields$^1$ and Antonino Marcian\`{o}$^2$}\\ \\
{\it$^1$ 23 Rue des Lavandi\`{e}res}\\
{\it 11160 Caunes Minervois, FRANCE}\\
{fieldsres@gmail.com}\\ \\
{\it$^2$ Center for Field Theory and Particle Physics \& Department of Physics} \\
{\it Fudan University, Shanghai, CHINA} \\
{marciano@fudan.edu.cn}
}
\maketitle

\begin{abstract}
We show that sharing a quantum reference frame requires sharing measurement operators that identify the reference frame in addition to operators that measure its state.  Observers restricted to finite resources cannot, in general, operationally determine that they share such operators.  Uncertainty about whether system-identification operators are shared induces decoherence.
\end{abstract}

\textbf{Keywords:} Classical communication; LOCC protocol; Measurement; Quantum Darwinism; Quantum reference frame

\section{Introduction}

Bartlett, Rudolph and Spekkens~\cite{bartlett:07} define \emph{nonfungible} quantum information as information that Alice and Bob can share only by exchanging a quantum (i.e., physical) system. This is in contradistinction to \emph{fungible} information that Alice can communicate to Bob by sending a string of classical bits, e.g., a bit string that describes a system or encodes a measurement outcome.  Quantum systems implementing reference frames (quantum reference frames or QRFs~\cite{aharonov:84}), e.g., physically-implemented length and angle standards, clocks, gyroscopes, and standardized charges, are canonical examples of physically-encoded nonfungible information.  {Angelo et al. have shown~\cite{angelo:11} that such systems must be described with care to prevent unnoticed classical assumptions, particularly assumptions of separability, from introducing paradoxes even when only a single experiment, observer, and QRF are considered.}

{In practice, we are concerned not only with the formal description of a QRF, but also with the use of a QRF by an observer to make a measurement.}  Much of~\cite{bartlett:07} and of the broader literature is dedicated to developing methods for sharing fungible information even in the absence of shared QRFs, {i.e., methods that assume each observer employs only local QRFs that are available \emph{a priori}}.  The~physical operations required to share a QRF remain largely neglected, and have yet to be rigorously~characterized.  

To serve as a QRF, a physical system must have a designated pointer state that conveys classical reference information, e.g., spatial orientation if the QRF is a Cartesian frame; we consider the ``state'' of the QRF to be this designated pointer state.  Alice and Bob can exchange a QRF only if Bob can unambiguously identify the physical system received from Alice and measure the same (pointer) state of that system that Alice has previously measured or prepared.  Here we consider the process of~sharing a QRF in an operational setting, considering in particular the operations by both parties that are required to both identify the QRF as a system and determine its state, e.g., its spatial orientation if it is a Cartesian frame.  Employing the methods and the results of~\cite{fields:18}, we characterize the process of~identifying a QRF explicitly in terms of the measurement operators, i.e., observables employed to distinguish the QRF from its surrounding environment, including the other systems present in the laboratory.  We then show that the physical implementations of these observables encode nonfungible information.  Pointer state outcomes specifying the relative states of QRFs, which as shown in~\cite{bartlett:07} can be inferred without physically sharing the QRFs, are defined only with respect to such nonfungible, measurement-operator encoded information.  Hence the nonfungible information encoded by the physical implementations of QRF-identifying observables must already be shared either to exchange a~QRF or to infer the relative states of non-shared QRFs.  

Whether Alice and Bob implement the same observables cannot be determined operationally with finite resources, either by Alice and Bob or by a third party.  Hence Alice and Bob must respect a~superselection rule or, equivalently, experience decoherence~\cite{bartlett:07} {(see~\cite{angelo:11} for an explicit analysis of such decoherence even in a setting involving only three particles)}.  We suggest that the shared ``classical reality'' of the laboratory emphasized by Bohr~\cite{bohr:37} can be attributed to this unavoidable uncertainty about shared observables.

\section{System Identification Formalism}

We consider the situation shown in Figure~\ref{fig.1}: Alice makes measurements at her location, then sends Bob both her measurement outcomes (as fungible information) and (a token of) her Cartesian frame, which Bob employs, together with his local Cartesian frame, to make measurements at his distant location.  Avoiding no-cloning restrictions requires that if this token is in a pure state, that state is distinct from the state of Alice's local Cartesian frame; we assume for simplicity that the shared token is in a mixed state, and justify this assumption below.  We also assume explicitly that the process of transferring the token to Bob does not change its state, then show in \S 5 below how uncertainty about the state can nonetheless be introduced by the sharing protocol.  We ask how both Alice and Bob \emph{identify} the transferred QRF as a physical system, i.e., distinguish it from the surrounding environment, including whatever else is contained within their respective laboratories, and how they then determine the pointer state of the shared reference frame, e.g., its orientations with respect to their respective local Cartesian frames.  Alice and Bob can share nonfungible reference-frame information by exchanging a~QRF token only if they both identify the same physical system as the shared token; if Bob receives, identifies, and measures the state of a different system from the one Alice identified, measured or prepared, and sent, the intended nonfungible information clearly has not been successfully shared. %please confirm if it should be Section 5?

\begin{figure}[h]
\centerline{\includegraphics[width=6.0in]{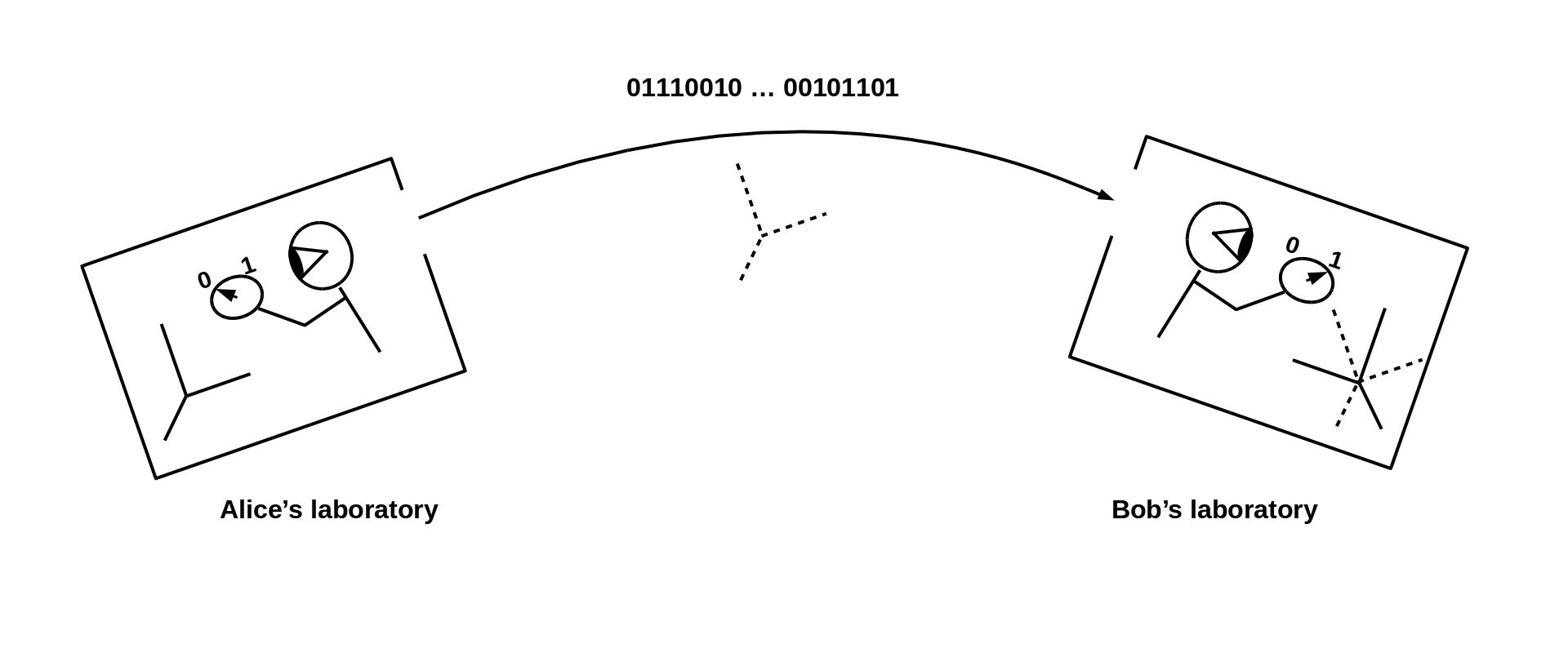}}
\vspace*{8pt}
\caption{Alice sends distant Bob a fungible encoding of her observational outcomes and a nonfungible token (dashed lines) of her local Cartesian frame.  Both Alice and Bob must identify the same token for the sharing protocol to be successfully executed.}
\label{fig.1}
\end{figure}

Letting $k = A$ or $B$  and following the methods of~\cite{fields:18}, we first consider partitions of ``everything'' $U$ into an observer $k$ and that observer's ``world'' $W_k$, i.e., everything with which that observer can interact.  For a non-relativistic system, we can write $U = AW_A = BW_B$ and we can consider Hilbert spaces $\mathcal{H}_U = \mathcal{H}_k \otimes \mathcal{H}_{W_k}$ and interaction Hamiltonians $H_U = H_k + H_{W_k} + H_{kW_k}$; relativity requires identifying $W_k$ with the union of $k$'s past and future lightcones, but does not change what follows.  Consider each observer to interact with its respective world by deploying Hermitian operators $M^k_i$, choosing bases in which these operators have binary eigenvalues and hence correspond to ``questions to Nature'' with yes--no answers~\cite{wheeler:83}.  We assume the observers have only finite resources and hence finite numbers of such operators; for simplicity, we assume a common total number of operators $N$.  In~all realistic situations, observers have limited knowledge of and operational control over the worlds with which they interact, i.e., $N \ll dim(W_k)$ for any observer $k$.  In this case, we can write, for~each~observer:

\begin{equation}
H_{kW_k} = \beta^k k_B T^k \sum_i \alpha^k_i M^k_i,
\end{equation}  
where the $\alpha^k_i \in [0,1]$ are such that $\sum_i \alpha^k_i = 1$, $k_B$ is Boltzmann's constant, $T^k$ is $k$'s temperature, and $\beta^k \geq$ ln 2 is a measure of $k$'s thermodynamic efficiency that depends on the internal dynamics $H_k$.  The idea that Nature ``answers'' an observer's questions is classical, and implies an irreversible state change~\cite{landauer:61}: Each question from $k$ that $W_k$ ``answers'' transfers one bit from $W_k$ to $k$ and is paid for by the transfer of $\beta^k k_B T^k$ from $k$ to $W_k$.  The action to transfer $N$ bits in time $\Delta t$ is:

\begin{equation}
\int_{\Delta t} dt (\imath \hbar) \mathrm{ln} \mathcal{P}_k(t) = N \beta^k k_B T^k \Delta t
\end{equation} 
where $\mathcal{P}_k = \exp{-(\imath /\hbar) H_{kW_k} t}$.

We can now ask how each observer allocates their limited observational resources to the tasks of~identifying and characterizing the pointer states of the various ``systems'' embedded in their respective worlds.  Consistent with the arbitrary tensor-product decomposability of $W_k$, we make no assumption that such systems are ``ontic'' but rather consider them to be defined by $k$'s activity of~measurement~\cite{fields:18, grinbaum:17}.  Let $X_k$ be the Cartesian QRF to be shared as identified by $k$ and $Y_k$ be $k$'s local Cartesian frame.  We can write $k$'s system-identification observables as $M^{Xk}_j$ and $M^{Yk}_j$, suppressing redundant subscripts $k$ and again assuming equal total numbers $M$, $2M < N$ of $X_k$- and $Y_k$-identifying operators for simplicity.  Similarly, $k$'s pointer observables for determining the state, e.g., orientation, of $X_k$ and $Y_k$ are $M^{Pk}_l$ and $M^{Qk}_l$, with a common number of operators $M^\prime < M$ for each pointer state.  We explicitly assume that the $M^{Pk}_l$ and $M^{Qk}_l$ act on whatever system is identified by the $M^{Xk}_j$ and $M^{Yk}_j$,~respectively.

To identify the shared QRF and determine its state with respect to their own local Cartesian frame, both Alice and Bob must execute the cycle of measurements shown in Figure~\ref{fig.2}.  Alice first identifies her local frame and measures its state, then identifies the token QRF to be transferred and either measures its state as, or prepares its state to be, indistinguishable within her measurement resolution from that of her local frame.  As Alice's measurement resolution is finite, this state will in general be~mixed.  Bob receives and identifies a token, which for successful communication must be the same one transferred, measures its state, and then identifies and measures the state of his local frame to make the comparison.  At each step in the cycle, all degrees of freedom not being measured in that step are part of the ``environment'' for the measurements being made~\cite{tegmark:12}.  This redefinition of~the~environment between measurements implements decoherence~\cite{fields:19}, a point we will return to in \S 5 below.
 %please confirm if it should be Section 5?

If it is assumed \emph{a priori} that Alice and Bob share both system-identification and pointer-state operators, i.e., if $\{ M^{XA}_j \} = \{ M^{XB}_j \}$ and $\{ M^{PA}_l \} = \{ M^{PB}_l \}$, then $X_A = X_B$ and, assuming as above that transmission of the QRF does not change its pointer state, $| X_A \rangle = |X_B \rangle$.  Hence if it is assumed \emph{a priori} that Alice and Bob share both system-identification and pointer-state operators, they can exchange nonfungible information by exchanging the QRF $X$.  From an operational perspective, however, this cannot be assumed.  The question of interest in this case is whether Alice and Bob can determine, by~finite observations, that they have successfully exchanged $X$, i.e., that $X_A = X_B$.  In particular, can Bob determine by finite observations that the physical system $X_B$ identified by his operators $M^{XB}_j$ is the same physical system $X_A$ that Alice identified using her operators $M^{XA}_j$?  Without additional information from Alice, clearly the answer is no.  Alice sending additional nonfungible information in the form of additional physical systems that require identification by Bob, moreover, merely leads to infinite regress.  Hence the operational question is: Can Bob determine by finite observations that $X_A = X_B$ given additional \emph{fungible} information from Alice.  We show below that the answer is no.

\begin{figure}[h]
\centerline{\includegraphics[width=3.0in]{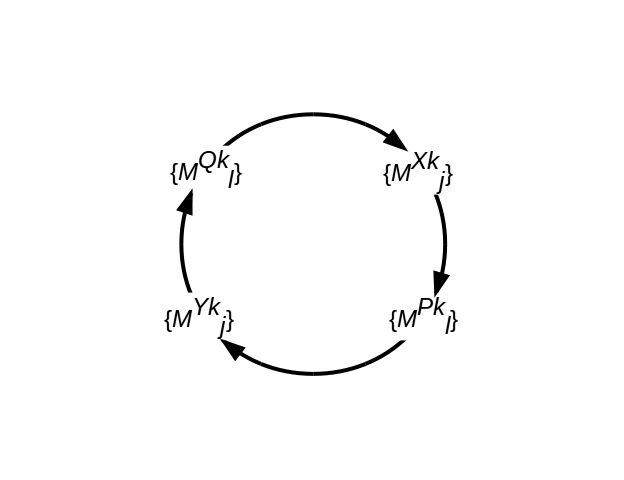}}
\vspace*{8pt}
\caption{Alice and Bob each cycle through identifying ($\{ M^{Xk}_j \}$ and $\{ M^{Yk}_j \}$) and then determining the pointer state of ($\{ M^{Pk}_l \}$ and $\{ M^{Qk}_l \}$) their shared ($X_k$) and local ($Y_k$) Cartesian frames.}
\label{fig.2}
\end{figure}

\section{Fungible Information Is Insufficient for System Identification}

Bob's problem is to determine, given his identified system $X_B$ and a fungible description of Alice's identified system $X_A$, whether $X_A = X_B$.  By ``fungible'' here and below we also understand ``obtained by finite observations.''  The following shows that this problem cannot be solved.

\begin{Theorem}
Given a system $X_B$, no fungible description of a system $X_A$ previously observed at some distant location is sufficient to determine whether $X_A = X_B$.
\end{Theorem}

To prove this, we focus on how the fungible description is constructed by the distant observer $A$.  The description of $X_A$ can be sufficient to determine whether $X_A = X_B$ only if, for any degree of freedom $\phi$ in $W_A$, it specifies whether $\phi$ is a degree of freedom of $X_A$.  To achieve this, $A$'s operators $M^{XA}_j$ must be sufficient to determine whether any given $\phi$ is a degree of freedom of $X_A$.  Let $\overline{X_A}$ be such that $X_A \overline{X_A} = W_A$. Abusing the notation slightly, we can write their action as $M^{XA}_j: X_A \rightarrow 1$ for all $j$ and $M^{XA}_j: \overline{X_A} \rightarrow 0$ for some $j$ (all $j$ if the $M^{XA}_j$ are a minimal nonredundant set).  Being able to identify $X_A$ at multiple times requires that the $M^{XA}_j$ satisfy the predictability sieve condition~\cite{zurek:03}:

\begin{equation} \label{sieve}
[M^{XA}_j, H_{W_A}] = 0~~\forall j.
\end{equation}
\nolinebreak
The system $X_A$ must, in other words, be separable from $\overline{X_A}$ both before and after $A$'s measurement, and the measurement must not disturb $W_A$ by more than the measurement resolution.  Equivalently, $A$ has an operator $M^{\overline{X}}$ such that $M^{\overline{X}}: \overline{X_A} \rightarrow 1$ and $M^{\overline{X}}: X_A \rightarrow 0$ that satisfies $[M^{XA}_j, M^{\overline{X}}] = 0~~\forall j$.  Determining whether any given $\phi$ is a degree of freedom of $X_A$ is then determining whether $M^{\overline{X}}(\phi) = 0$ or $1$.

\begin{Lemma}
An observer $A$ cannot determine, by finite observation, whether any arbitrary degree of freedom $\phi$ is a degree of freedom of a specified system $X_A$.
\end{Lemma}

\begin{proof}[Proof of Lemma 1]
We take the proof from that of~\cite{fields:18} Theorem 1.  To determine whether any arbitrary degree of freedom $\phi$ is a degree of freedom of the specified system $X_A$, $A$ must, in the limit, examine every potential $\phi$, i.e., every degree of freedom of $W_A$.  This requires progressively refining the observation of $W_A$ by employing additional measurement operators, with the number of operators $N \rightarrow dim(W_A)$ in the limit.  In this limit, Equation (2) becomes:

\begin{equation}
\int_{\Delta t} dt (\imath \hbar) \mathrm{ln} \mathcal{P}_k(t) \rightarrow dim(W_A) \beta^k k_B T^k \Delta t.
\end{equation}

Now consider two measurements made at $t$ and at $t + \Delta t$, after the heat given by Equation (4) from the measurements initiated at $t$ has been dumped into $W_A$.  However this heat is distributed in $W_A$, the predictability sieve condition (3) will be violated for at least one of the operators $M^{XA}_j$ at $t + \Delta t$.  It is, however, Equation (3) {that assures that $X_A$ is in a separable state and hence} allows $A$ to identify $X_A$.  {If Equation (3) and hence separability fail,} $A$'s attempt to re-identify $X_A$ at $t + \Delta t$ will fail as well.  In this case, $A$ cannot determine whether $\phi$ is a degree of freedom of $X_A$.
\end{proof}
Lemma 1 shows that $A$ cannot determine, by finite observation, the value that the physically implemented operator $M^{\overline{X}}$ assigns to any arbitrary $\phi$.  Equivalently, $A$ cannot determine by finite observation that the operators $M^{XA}_j$ satisfy Equation (3), as doing so requires full operational control over $H_{W_A}$.  Hence $A$ cannot determine by finite observation that the $M^{XA}_j$ pick out the same system $X_A$ at successive times, as indeed follows from even a classical analysis of system identification under finite-resource constraints (\cite{moore:56} Theorem 2; see~\cite{fields:18} for discussion).

Given Lemma 1, proving Theorem 1 is trivial:

\begin{proof}[Proof of Theorem 1]
By Lemma 1, no fungible description of $X_A$ is sufficient to specify the degrees of freedom of $X_A$.  Hence no such description is sufficient to determine that $X_A = X_B$, for any $X_B$.
\end{proof}

Lemma 1 clearly applies to Bob as well as to Alice; Bob is equally unable to determine, by finite observation, the degrees of freedom of the system $X_B$ that he has received.  Hence even given an \emph{a priori} stipulation, instead of a description from observation, of the degrees of freedom of $X_A$, Bob cannot determine that $X_A = X_B$.  This situation is clearly no different if Alice and Bob both receive tokens, with or without accompanying descriptions, from some third party.  Unless it is assumed \emph{a priori} that Alice and Bob share both system-identification and pointer-state operators, they face unresolvable uncertainty about whether they are in fact identifying the same physical system $X$ or measuring the same pointer state $|X \rangle$, i.e., they face unresolvable uncertainty about whether they in fact share a QRF.  In the language of~\cite{bartlett:07}, this is ``passive'' uncertainty.  It cannot, however, be distinguished operationally from ``active'' uncertainty about whether the unobserved (or discontinuously observed) process of transferring a single, well-defined token $X$ from Alice to Bob changed its state.

\textbf{Example 1:}  System identification is typically treated as unproblematic in discussions of QRF sharing (e.g.,~\cite{bartlett:07}).  Consider, however, an adversarial situation in which a malicious third party intercepts the transferred QRF and substitutes a distinct physical system $X^\prime$ for $X_A$.  How much fungible information must Alice provide to assure that Bob can detect the substitution?  If Alice and Bob already share QRFs, instructions to perform some set of measurements with respect to the already-shared QRFs are sufficient.  In the device-independent quantum key distribution protocol of~\cite{vazirani:14}, for example, Alice and Bob share \emph{a priori} a Cartesian frame with respect to which Bell tests can be made.  This protocol clearly fails if Alice and Bob share no QRFs \emph{a priori} and the adversary is allowed to intercept and manipulate any QRFs they attempt to exchange.

\textbf{Example 2:}  Transferring a qubit is transferring nonfungible information.  Using a transferred qubit to encode information requires a previously-shared QRF, e.g., a Cartesian frame for spin measurements.  ``Direct'' communication protocols in which qubits are serially transferred~\cite{qui:11, qui:14} therefore require previously-shared QRFs.  Similar considerations apply when tranferring a qubit in time (see~\cite{fields:18} for details).  Either multiple, serial measurements or multiple, serial preparations (e.g.,~\cite{qui:15}) of a qubit require an \emph{a priori} time-invariant QRF that implements a fixed measurement and/or preparation basis.

\textbf{Example 3:}  Any fungible information sent by Alice to Bob must be physically encoded~\cite{landauer:99}.  The physical system encoding the fungible information must itself be exchanged, so is effectively a~QRF.  To access the encoded fungible information, the recipient must measure the state of this system, as noted already in the case of the classical communication step in Bell measurements~\cite{tipler:14}.  Hence QRF-exchange protocols that assume independent fungible information exchange are technically circular.  Identifying an encoded message as a message is nontrivial in the absence of an \emph{a priori} shared language, as the celebrated \emph{Voyager} probe communications of 1977 exemplify.  A shared language, however, can also be considered a shared QRF, as shown in the special case of measurement operators below.

\section{Implemented System Identification Operators Are Nonfungible}

Observers identify physical systems by deploying system-identification operators.  Such operators are \emph{implemented} by the physical structures of the observers.  Given Lemma 1, it is easy to see the~following:

\begin{Theorem}
An observer cannot determine, by finite observation, what operators identify a physical system $X$.
\end{Theorem}

\begin{proof}
Suppose the opposite: That $A$ can determine by observation that some finite set of operators $\{ M_j \}$ identify $X$.  A formal specification of the set $\{ M_j \}$ of operators is then, operationally, a fungible description of $X$.  By Lemma 1, however, no fungible description of $X$ is sufficient to specify the degrees of freedom of $X$.  Hence the operators $\{ M_j \}$ do not, in fact, identify $X$, contradicting our~assumption.
\end{proof}

Alternatively, Lemma 1 shows that no finite set of observations of an observer is sufficient to determine the degrees of freedom of the observer.  Hence no such set of observations is sufficient to determine what set of system-identification operators the observer's physical structure implements.

This result is in fact obvious: Alice can send fungible descriptions of her operators $\{ M^{XA}_j \}$ and $\{ M^{PA}_l \}$ to Bob, but she cannot send her operators \emph{as physically implemented} to Bob, as this would require sending herself, or some component of herself such as her brain.  Hence \emph{implemented} system identification operators are nonfungible; they cannot be transferred from one observer to another by sending a finite bit string.  They are effectively internal QRFs, inseparable from the physical structure of the observer that deploys them.  Alice's and Bob's passive uncertainty about whether they share system-identification operators as internal QRFs is the source of their passive uncertainty about whether they share any external QRFs.

\section{Uncertainty about Operator Sharing Induces Decoherence}  \label{sec.5}

As shown in~\cite{bartlett:07}, uncertainty about whether a QRF is shared induces decoherence.  Briefly, even if Alice prepares the transferred QRF token in a pure state, Bob cannot determine that it is pure; his~uncertainty about whether his operators $\{ M^{XB}_j \}$ and $\{ M^{PB}_l \}$ pick out the same system as Alice's operators $\{ M^{XA}_j \}$ and $\{ M^{PA}_l \}$ is equivalent to uncertainty about system preparation.  Bob's measured pointer state $|X_B \rangle$ with respect to his local frame is, therefore, mixed even if Alice's preparation of $|X_A \rangle$ is pure with respect to her local frame.

The local operational indistinguishability between active and passive uncertainty about reference-frame sharing renders environmental and reference-frame induced decoherence locally operationally indistinguishable.  To see this, consider the situation as represented by quantum Darwinism~\cite{zurek:06, zurek:09}.  Here multiple observers interact with multiple partitions $E^k$ of the environment $E$ of some system $S$.  Environmental decoherence acting at the $S-E$ boundary redundantly encodes the eigenvalues of the interaction $H_{SE}$ into each of the $E^k$; the observers obtain this information by mutually-nondisturbing interactions with their respective $E^k$ as sketched in Figure~\ref{fig.3}.  Agreement among the observers about the state of $S$ is enforced by the redundant encoding, which is in turn enforced by the ``einselection'' of the eigenvalues of $H_{SE}$ as the only stably-encodable classical information.  Stability is defined by a prediction sieve requirement, $[H_S + H_{SE}, P_\psi] = 0$~\cite{zurek:03}, where $P_\psi$ projects the stable, and hence observable, state $\psi$ of $S$.  

\begin{figure}[h]
\centerline{\includegraphics[width=4.0in]{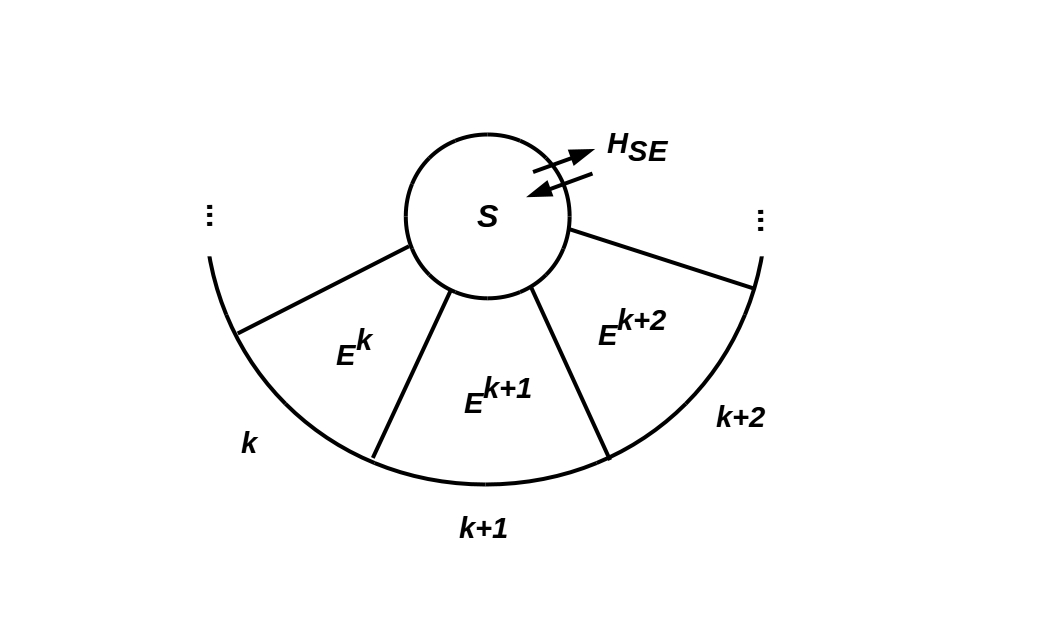}}
\vspace*{8pt}
\caption{Quantum Darwinism~\cite{zurek:06, zurek:09}: Observers $k$ interact with disjoint partitions $E^k$ of the environment $E$ of $S$, each of which encodes the eigenvalues of $H_{SE}$.  Redundancy of encoding and hence agreement among the observers $k$ clearly requires that $H_{SE}$ be independent of the partitioning of $E$ into the $E^k$.}
\label{fig.3}
\end{figure}

In this situation, the environmental partitions $E^k$ each serve as QRFs that their respective observers $k$ employ to measure the state of $S$.  As the observers are by assumption mutually isolated, they share no other QRFs for their measurements of $S$.  As no pair of observers shares a reference frame, they are limited to the exchange of fungible information, i.e., formal specifications of their local interactions $H_{SE^k}$ and observational outcomes.  The latter will agree if the $H_{SE^k}$ are uniform within the measurement resolution, i.e., provided the eigenvalues of $H_{SE}$ are independent of the partitioning of $E$ into the $E^k$.  If~this requirement is violated, the predictability sieve fails, the encoding of eigenvalues is not redundant, and the observers $k$ have no basis for claiming to observe the same system $S$.  The observers cannot, however, operationally determine that the eigenvalues of $H_{SE}$ are independent of the partitioning of $E$ while each restricted to their own partition, and hence QRF $E^k$~\cite{fields:10}.  They cannot, therefore, operationally determine that they are observing the same state $\psi$ or even the same system $S$.

\section{Conclusions}

It is a standard assumption of all empirical science that observers are interchangeable.  Implicitly, this is an assumption that observers in fact implement the same measurement operators, and hence the same ``internal'' QRFs.  If this is the case, they share nonfungible information \emph{a priori}, and hence can share additional nonfungible information implemented by exchangable QRFs.  

Finite observers cannot, however, operationally determine that they implement the same measurement operators, and hence cannot operationally determine that they are interchangable.  They cannot, therefore, operationally determine that they share an exchangable QRF.  This uncertainty induces decoherence.  We suggest that the ``shared classical world'' that must be assumed to describe experiments as independently reproducible results from this decoherence.

%%%%%%%%%%%%%%%%%%%%%%%%%%%%%%%%%%%%%%%%%%
\section*{Acknowledgements}
This research was funded by the Federico and Elvia Faggin Foundation (CF), the Shanghai Municipality, Grant No. KBH1512299 (AM), and by Fudan University, Grant No. JJH1512105 (AM).

\end{document}